\documentclass[
final
]{dmtcs-episciences}

\usepackage[utf8]{inputenc}
\usepackage{subfigure}
\usepackage{thmtools}
\usepackage{thm-restate}
\usepackage{amsmath}
\usepackage{mathtools}
\usepackage{graphicx}
\usepackage{caption}
\usepackage{float}

\newtheorem{thm}{Theorem}[section]
\newtheorem{lem}[thm]{Lemma}
\newtheorem{conj}[thm]{Conjecture}
\newtheorem{defn}[thm]{Definition}
\newtheorem{prop}[thm]{Proposition}

\newtheorem{remark}[thm]{Remark}
\newtheorem{cor}[thm]{Corollary}

\usepackage[round]{natbib}

\newcommand{\R}{\mathbb{R}}
\newcommand{\Ag}[1]{A_{\text{gen}}^{(#1)}}
\newcommand{\floor}[1]{\lfloor #1 \rfloor}

\author[Ascoli et al.]{Ruben Ascoli\affiliationmark{1}\thanks{This work was supported by NSF grant DMS1947438 and Williams College.} \and Livia Betti\affiliationmark{2} \and Jacob Lehmann Duke\affiliationmark{3} \and Xuyan Liu\affiliationmark{4}\thanks{This author was additionally supported by the University of Michigan.} \\ \and Wyatt Milgrim \affiliationmark{5} \and Steven J. Miller\affiliationmark{3} \and Eyvindur A. Palsson\affiliationmark{6}\thanks{This author was supported in part by the Simons Foundation grant $\#$360560.} \\ \and Francisco Romero Acosta\affiliationmark{6} \and Santiago Velazquez Iannuzzelli\affiliationmark{7}}

\title[Distinct Angles and Angle Chains in Three Dimensions]{Distinct Angles and Angle Chains in Three Dimensions}
\affiliation{
  Princeton University \\
  University of Rochester \\
  Williams College \\
  University of Michigan \\
  Vassar College \\
  Virginia Tech \\
  University of Pennsylvania
  }
\begin{document}
\publicationdata{vol. 25:1}{2023}{2}{10.46298/dmtcs.10037}{2022-09-13; 2022-09-13; 2023-02-07}{2023-02-09}
\maketitle
\begin{abstract}
  In 1946, Erd\H{o}s posed the distinct distance problem, which seeks to find the minimum number of distinct distances between pairs of points selected from any configuration of $n$ points in the plane. The problem has since been explored along with many variants, including ones that extend it into higher dimensions. Less studied but no less intriguing is Erd\H{o}s’ distinct angle problem, which seeks to find point configurations in the plane that minimize the number of distinct angles. In their recent paper “Distinct Angles in General Position,” Fleischmann, Konyagin, Miller, Palsson, Pesikoff, and Wolf use a logarithmic spiral to establish an upper bound of $O(n^2)$ on the minimum number of distinct angles in the plane in general position, which prohibits three points on any line or four on any circle.
  
  We consider the question of distinct angles in three dimensions and provide bounds on the minimum number of distinct angles in general position in this setting. We focus on pinned variants of the question, and we examine explicit constructions of point configurations in $\mathbb{R}^3$ which use self-similarity to minimize the number of distinct angles. 
Furthermore, we study a variant of the distinct angles question regarding distinct angle chains and provide bounds on the minimum number of distinct chains in $\mathbb{R}^2$ and $\mathbb{R}^3$. 
\end{abstract}

\section{Introduction}

Erd\H{o}s' distinct distance problem, introduced in \cite{Erd}, is one of the most famous problems in discrete geometry. If $g(n)$ is the minimal number of distinct distances among $n$ points in the plane, Erdős conjectured that the $\sqrt n \times \sqrt n$ lattice, yielding $g(n)=O(n/\sqrt{\log n})$, is the optimal configuration. 
The problem was almost completely resolved by \cite{GuKa}, who proved a lower bound of $g(n)=\Omega(n/\log n)$.

Many variants of this problem and related ones have since arisen. 
The question of distinct distances in three and higher dimensions is studied by \cite{SoVu}, and previously by \cite{ArPa} and \cite{ClEd}. Another variant relevant to our paper regards chains of distances. \cite{Pass} considers the question of the minimum number of distinct $k$-tuples of distances. Specifically, for a configuration of points $\mathcal P$, he defines 

\[\Delta_k(\mathcal P) = \{(|p_1-p_2|, |p_2-p_3|, \ldots, |p_k-p_{k+1}|): p_i\in \mathcal P\},\] and he finds a lower bound on $|\Delta_k(\mathcal P)|$ for any configuration $\mathcal P$.


Our work concerns a far less-studied question proposed by \cite{ErPu} concerning the number of distinct angles among $n$ not all colinear points in the plane. They conjectured that the regular polygon is the optimal construction, with $n-2$ distinct angles, and obtain a lower bound of $(n-2)/2$ under the additional restriction that no three points are colinear. \cite{FlHu} consider many variants of the question, finding many of them to be easily solved up to constant factors. One interesting variant they introduced––the one most relevant to this paper––requires the points to be in \emph{general position}, which they define to mean no subset of three points are colinear and no subset of four points are cocircular. (The problem of distinct distances has also been studied under the restriction of general position; see \cite{Dum}.) 

We adapt the following notation from \cite{FlHu}:
\begin{defn}
Denote by $\Ag{d}(n)$ the minimum number of distinct angles formed by a set of $n$ points in general position in $\R^d$.
\end{defn}

\cite{FlHu} showed that $\Ag{2}$ is lower bounded by $\Omega(n)$ (a bound also known to Erd\H{o}s), and they also obtained an upper bound of $\Ag{2} = O(n^{\log_2 7})$. This upper bound is improved to $\Ag{2} = O(n^2)$ by \cite{FlKo}. We discuss these results in detail in Section \ref{prev}.

In this paper, we consider the question of distinct angles in general position in $\R^3$ and obtain lower and upper bounds on $\Ag{3}$. Here, we still say that a configuration of points is in general position if it contains no three points on a line and no four points on a circle; this turns out to be a natural definition even in three dimensions, though it is less restrictive than the classical definition of general position, which would prevent four points on a plane or five points on a sphere. See Section \ref{genpos} for further discussion on the general position restriction.

Our main results come from considering pinned variants of the question, where we fix (pin) certain points and consider only angles that contain those points, often specifying whether they are endpoints or center points. An analogous question has been studied for distances, asking what is the minimum number of distances determined with any one point in the configuration. The conjectured answer in the case of the plane is the same as for the unpinned question. \cite{Erd} obtained the first lower bound of $\Omega(\sqrt{n})$; the current best lower bound was obtained by \cite{KaTa} (see Corollary 6 there). When the points are required to form a convex polygon, Erd\H{o}s conjectured that there is a point that determines $\floor{n/2}$ distinct distances; the best lower bound for this problem is given by \cite{NiPa}.

We also consider a question analogous to that considered by \cite{Pass} for distances. We find lower and upper bounds for the minimum number of distinct $k$-tuples of angles for configurations of points in general position. A similar question on angle chains, studying how many times a particular $k$-tuple of angles could occur, is considered by \cite{PSW}.

\subsection{Main Results}

 We first consider all possible angles formed by two pinned points. For this variant, we get the following result:

\begin{restatable}{thm}{twopoints}
\label{thm1}
In a configuration of $n$ points in general position in three dimensions, fix two points $A$ and $B$. The number of distinct angles with $A$ and $B$ as two of the three points forming the angle is at least $\sqrt{(n-2)/3}$. Furthermore, this lower bound is tight up to a constant: it is possible to have only $2\sqrt{(n-2)/3}-1$ such angles.
\end{restatable}
That is, the minimum number of distinct angles formed with two pinned points is $\Theta(\sqrt{n})$.

We then proceed to investigate the situation when only the center point is fixed. We reduce this problem to the question of distinct distances on a sphere, which gives us the following:
\begin{restatable}{thm}{pinnedcenter}\label{thm2}
Consider a configuration of $n$ points in general position in three dimensions, and pin a point $A$. The minimum number of angles formed with $A$ as the center point is $O(n)$ and $\Omega\left(n/\log n\right)$.
\end{restatable}

The last pinned point variant we consider is that arising from pinning an endpoint $A$. No nontrivial upper bound is known for this variant, but we prove the following nontrivial lower bound:
\begin{restatable}{thm}{pinnedend}\label{thm3}
Consider a configuration of $n$ points in general position in three dimensions, and pin a point $A$. The minimum number of angles formed with $A$ as an endpoint is $\Omega(\sqrt{n})$.
\end{restatable}

Theorems \ref{thm1}, \ref{thm2}, and \ref{thm3} are proved in Section \ref{pinned variants}. 

For the question of distinct angles with no points pinned, we conjecture the following:

\begin{conj}\label{conj1}
Both $\Ag{2}(n)$ and $\Ag{3}(n)$ are $\Theta(n^2)$.
\end{conj}

In Section \ref{constructions}, we consider two new non-planar point configurations in general position in $\R^3$ that have $O(n^2)$ distinct angles.
Thus Conjecture \ref{conj1} essentially states that these constructions are optimal up to constant factors. Both of the constructions exhibit \emph{self-similarity}, a property that we define precisely in Section \ref{constructions}. We conjecture that a configuration of $n$ points with the smallest possible number of angles must possess this property.

In Section \ref{angle chains}, we explore the question of distinct chains of angles in both $\R^2$ and $\R^3$. We denote $L_k^{(d)}(n)$ to be the minimum number of distinct $k$-tuples of angles with an associated chain of $k+2$ points forming those angles. Here the minimum is taken across all configurations of $n$ points in general position in $d$ dimensions.
For $d=2$, we prove the following two results.
\begin{restatable}{thm}{lowertwod}\label{lowerbound2d}
$L_k^{(2)}(n)=\Omega (\Ag{2}(n) \cdot n^{k-1})$. In particular, since $\Ag{2}(n) = \Omega(n)$, we have $L_k^{(2)}(n)=\Omega(n^k)$.
\end{restatable}

\begin{restatable}{thm}{uppertwod}\label{upperbound2d}
$L_k^{(2)}(n) = O (n^{k+1})$.
\end{restatable}
The gap between our upper bound and lower bound for $L_k^{(2)}(n)$ is precisely the gap between the upper bound and lower bound for $\Ag{2}(n)$. This is due to the nature of the proof by induction on both upper and lower bounds. 

In three dimensions, the question of distinct angle chains becomes much more difficult. We establish the following weaker lower bound:

\begin{restatable}{thm}{lowerthreed}\label{lowerbound3d}
In three dimensions,  \[L_k^{(3)}(n)=\begin{cases} \Omega\left(\frac{n^{(k+2)/3}}{(\log n)^{(k+2)/3}}\right) & \text{ if } k=1\bmod 3; \\ 
\Omega\left(\frac{n^{(k+1)/3}}{(\log n)^{(k-2)/3}}\right) & \text{ if } k=2\bmod 3; \\
\Omega\left(\frac{n^{k/3+1/2}}{(\log n)^{k/3}}\right) & \text{ if } k=0\bmod 3.\end{cases}\]
\end{restatable}
This result relies heavily on Theorem \ref{thm2} and Theorem \ref{thm3} to decompose an angle chain into independent angles, sometimes with a pinned endpoint. In particular, improvements to Theorem \ref{thm3} or on the lower bound on $\Ag{3}$ would immediately yield improvements to Theorem \ref{lowerbound3d}.

Theorems \ref{lowerbound2d}, \ref{upperbound2d}, and \ref{lowerbound3d} are proved in Section \ref{angle chains}.

Next, in Section \ref{loosening}, we discuss what happens when we loosen the restrictions of general position. Finally, in Section \ref{future}, we discuss possible directions for future research.

\subsection{Previous Work: Distinct Angles in Two Dimensions}\label{prev}

\cite{FlHu} introduce the problem of distinct angles in general position in $\R^2$ and discuss a lower bound of $\Ag{2}(n) = \Omega(n)$. \cite{FlKo} achieve an upper bound of $\Ag{2}(n) = O(n^2)$. These results form the basis for much of our work in three dimensions, so we discuss them in detail here.

The first part of the following result was known by Erd\H{o}s and is addressed by \cite{FlHu}; see Lemma 2.7 there. We give a proof here for completeness and also discuss the case of two pinned endpoints.
\begin{lem}\label{lem1}
The number of angles formed by a set of $n$ points in general position in $\R^2$ with either a fixed endpoint and middle point or with two fixed endpoints is at least $(n-2)/2$.
\end{lem}
\begin{proof}
First consider a fixed endpoint $A$ and middle point $B$, and let the other endpoint $C$ vary. A given angle of this kind may occur at most twice, for there are two lines passing through the point $B$ that form that angle with the line $\overline{AB}$ (or one line in the case of the right angle). No more than one point besides $B$ can be on each of these lines since there are no three points on a line. Thus at least $(n-2)/2$ angles are formed.

Now instead fix two endpoints $A$ and $B$ and arbitrarily choose a middle point $C$. Since the three points are not collinear, they lie on a unique circle. The inscribed angle theorem gives us that the angle $\angle ACB$ is exactly one half of the arc of this circle between $A$ and $B$ that does not pass through $C$.

Consider the collection of such arcs formed by varying $C$ though the $n-2$ points besides $A$ and $B$ in the set. Since we can have no four points on a circle, the circles formed are all distinct. Each arc measure may only occur at most twice: once on either side of the two points. Thus the number of angles formed is at least $(n-2)/2$.

See Figure \ref{linescircles2d} for accompanying images.
\end{proof}


\begin{remark}
Lemma \ref{lem1} fails in three dimensions. See Section \ref{conestoruses}.
\end{remark}

\begin{figure}[htbp]
    \begin{center}
    \subfigure[]{
    \includegraphics[scale=0.55]{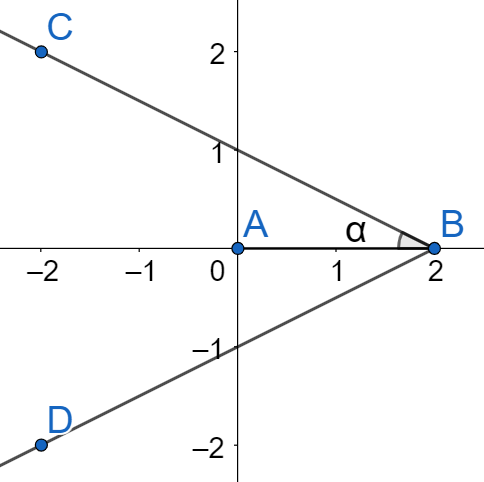}}
    \subfigure[]{
    \includegraphics[scale=0.5]{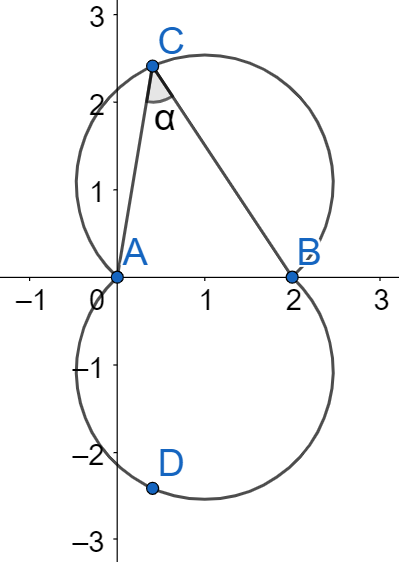}}
    \caption{When points are in general position, a given angle can only occur twice with two specified fixed points.}
    \label{linescircles2d}
    \end{center}
\end{figure}
\vspace{-3mm}
\cite{FlKo} obtain an upper bound of $\Ag{2}(n) = O(n^2)$ by distributing points along the logarithmic spiral $r=e^\theta$. They use the point set given in polar coordinates by $\mathcal P = \{p_j=(e^{\beta j}, \beta j):  j\in\{1,2,\ldots,n\}\}$, where $\beta$ is a small constant. \cite{FlKo} prove that an angle $\angle p_{j_1}p_{j_2}p_{j_3}$ is equivalent to an angle of the form $\angle p_{j_1+c}p_{j_2+c}p_{j_3+c}$ for any constant $c$. One can choose $c = 1-\min(j_1,j_2,j_3)$, showing that any angle on the logarithmic spiral can be formed using the point $p_1$. Then, there are $\binom{n-1}{2}$ choices for the other two points, yielding $3\binom{n-1}{2} = O(n^2)$ total distinct angles. See Figure \ref{logspiral}.

\begin{figure}[htbp]
    \centering
    \includegraphics[scale=0.45]{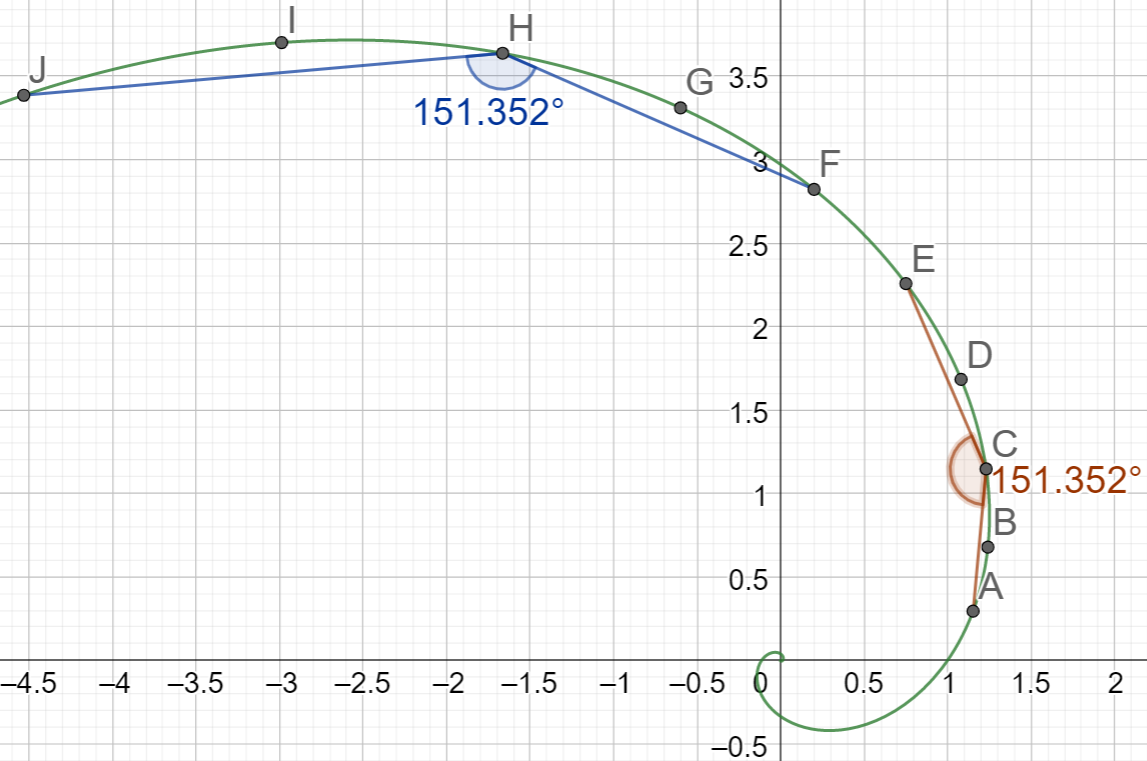}
    \captionsetup{aboveskip=0pt}
    \caption{The logarithmic spiral construction of \cite{FlKo}.}
    \label{logspiral}
\end{figure}

\section{Pinned Variants}\label{pinned variants}
\subsection{Cones and Spindle Tori}\label{conestoruses}
We now discuss why Lemma \ref{lem1} fails in three dimensions. Consider pinning, or fixing, the middle point and one endpoint of an angle and asking the following question: For the $n-2$ choices of the second endpoint, how many of the angles formed must be distinct? In two dimensions, Lemma \ref{lem1} told us that the answer is at least $(n-2)/2$. In three dimensions, however, all of the angles can be the same. To see why this is true, label the pinned endpoint $A$ and the pinned middle point $B$. Fix an angle $\alpha$ and form a ray with endpoint $B$ that has angle $\alpha$ with ray $BA$. Then, rotate the new ray around line $BA$. This forms a single-cone with vertex $B$ and axis $BA$. All points $C$ on the cone have the property that $m\angle CBA = \alpha$, so by distributing the remaining $n-2$ points on this cone (being careful not to place any three points on a line or any four points on a circle), all the angles will be the same. (Note that if $\alpha = \pi/2$, the object created is not a cone but rather a plane. Furthermore, if $\alpha > \pi/2$, the single-cone opens away from $A$, but still has line $BA$ as its axis. Neither of these observations make any difference in how this is used to prove our results.)

What if instead we pin the two endpoints of the angle and ask how many of the $n-2$ choices of center point must result in distinct angles? Once again, in two dimensions Lemma \ref{lem1} tells us that the answer is at least $(n-2)/2$, but in three dimensions, all of the angles can be the same. Label $A$ and $B$ as the endpoints and fix an angle $\alpha$. Choose a point $C$ such that $m\angle ACB = \alpha$. Consider the circle determined by the three points $A$, $B$, and $C$. Note that if one moves $C$ along arc $ACB$, this does not change the measure of angle $ACB$ since the angle is determined only by the measure of arc AB.

Now, rotate the circle formed by points $A$, $B$, and $C$ about the line $AB$. Rotating a circle about a line always forms a torus, but since the line in question passes through the circle (that is, intersects the circle twice), we specifically obtain a spindle torus. We do not actually want the entire spindle torus; in fact we only want to rotate arc $ACB$ about the line segment $AB$, giving us either the outer part of the spindle torus (if $\alpha < \pi/2$) or the inner part (if $\alpha > \pi/2$). If $\alpha = \pi/2$, then segment $AB$ is actually a diameter of the circle in question, and the rotation just gives us a sphere. In any of these three cases, for any point $C$ on the object that we form (which we henceforth just call ``spindle torus" even though we only have half of the full torus), we have $m\angle ACB = \alpha$.

See Figure \ref{conespindle} for accompanying images.

\begin{figure}[htbp]
    \begin{center}
    \subfigure[Place many points on a cone to avoid new angles with $A$ as an endpoint and $B$ as the center point.]{\includegraphics[scale=0.5]{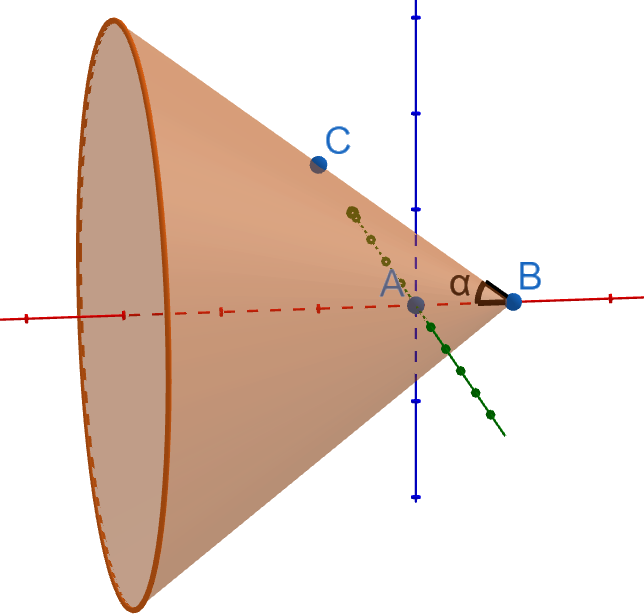}}
    \hspace{1cm}
    \subfigure[Place many points on a spindle torus to avoid new angles with $A$ and $B$ as the endpoints.]{\includegraphics[scale=0.48]{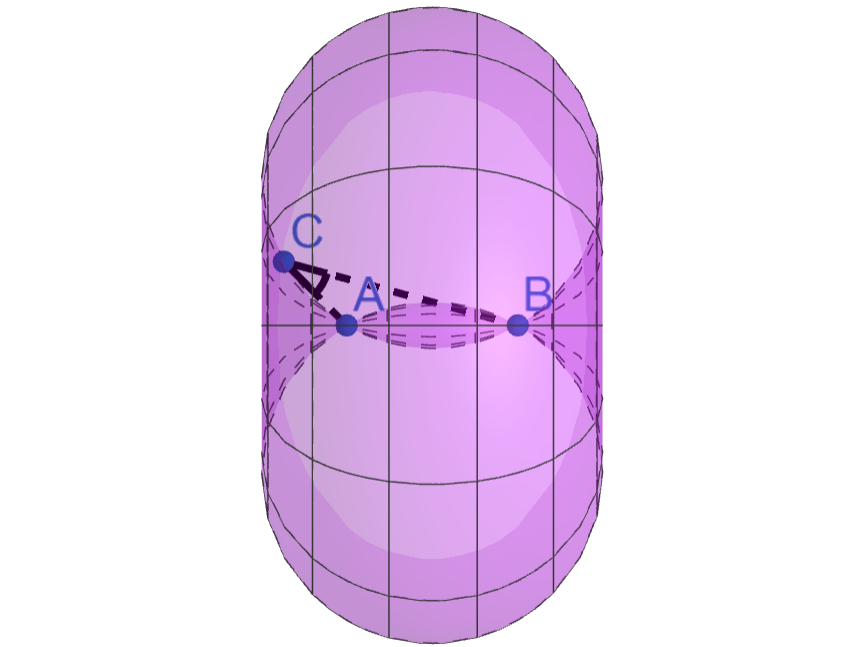}}
    \caption{Cone and Spindle Torus.}
    \label{conespindle}
    \end{center}
\end{figure}

\subsection{Pinning Two Points}
Recall that in Lemma \ref{lem1} for two dimensions, we fixed two points either as the two endpoints, or one point as an endpoint and one point as the middle point. We showed that in either of these two cases, the number of angles formed is $\Omega(n)$, but we discussed that in three dimensions, in either of these two cases it is possible to get only $O(1)$ distinct angles. Here, however, we prove Theorem \ref{thm1}, showing that if we fix two points $A$ and $B$ and consider \emph{all} the angles involving those two points---that is, counting all three cases of whether $A$ is the middle point, $B$ is the middle point, or both are endpoints---then the minimum number of distinct angles is $\Theta(\sqrt{n})$. For convenience, we repeat the precise statement of Theorem \ref{thm1}.

\twopoints*
\begin{proof}
Consider a third point $C$. There is a unique single-cone that contains the point $C$, has $A$ as the vertex, and has line $AB$ as its axis. All angles formed by a point on this cone as an endpoint, $A$ as the middle point, and $B$ as the other endpoint are equivalent. Consider also the unique single-cone that contains $C$, has $B$ as the vertex rather than $A$, and once again has line $AB$ as its axis. All angles formed by a point on this cone as an endpoint, $B$ as the middle point, and $A$ as the other endpoint are equivalent.

For each point besides $A$ and $B$, form the two cones described above. Each point besides $A$ and $B$ lies on the intersection of one of the cones with vertex $A$ and one of the cones with vertex $B$ that are constructed in this manner. Notice, however, that since the two cones have the same axis, this intersection is a circle. There cannot be four points on a circle, so for each pair of cones, there can only be three points on this intersection.

This means that if $x$ distinct cones with vertex $A$ are formed and $y$ distinct cones with vertex $B$ are formed, there can only be $3\cdot x\cdot y+2$ total points; this is because there are $x\cdot y$ pairs of cones, and the $+2$ are the points $A$ and $B$. Said another way, $(n-2)/3\leq x\cdot y$.

We now count the total number of distinct angles formed with points $A$ and $B$, in terms of $x$ and $y$. From taking a point on each cone, we automatically get $x$ distinct angles of the form $\angle CAB$, and we get $y$ distinct angles of the form $\angle CBA$. These might overlap, but we have at least $\max(x,y)$ distinct angles. 


To minimize this quantity while also satisfying $xy\geq (n-2)/3$, we let $x=y=\sqrt{(n-2)/3}$. Thus the number of distinct angles is at least $\sqrt{(n-2)/3}$, as desired.

Following this proof, not much new needs to be done to show tightness. However, we do also need to consider the angles of the form $\angle ACB$. 

Assume that $(n-2)/3$ is an odd perfect square. Fix points $A$ and $B$; then, fix $\sqrt{(n-2)/3}$ cones with $A$ as the vertex and line $AB$ as the axis and another $\sqrt{(n-2)/3}$ cones, each congruent to a cone in the first set, with $B$ as the vertex and line $AB$ as the axis. Choose the cones to have angles between, say, $5\pi/18$ and $7\pi/18$ degrees inclusive with their axis, distributed in an arithmetic progression. (Note: since $\sqrt{(n-2)/3}$ is odd, $\pi/3$ is included.) For each choice of one cone with vertex $A$ and one cone with vertex $B$, their intersection forms a circle; place three points on each of these circles. There are $\sqrt{(n-2)/3}\cdot \sqrt{(n-2)/3} = (n-2)/3$ choices of two cones, so in this manner we have placed all $n-2$ points that are not $A$ or $B$. Furthermore, the points are in general position if we choose locations on each circle wisely to avoid three points on a line or four on a circle. The number of distinct angles of the form $\angle CAB$ or $\angle ABC$ is $\sqrt{(n-2)/3}$ (since the angles are the same for the two sets of cones). The number of distinct angles of the form $\angle ACB$ is precisely $2\sqrt{(n-2)/3}-1$, and these angles include the $\sqrt{(n-2)/3}$ counted previously. This is because the angles $\angle ACB$ are the arithmetic progression from $4\pi/18$ to $8\pi/18$ with the same common difference as the angles in the original arithmetic progression. Since $\pi/3$ was in the original arithmetic progression, $5\pi/18$ and $7\pi/18$ are in the new arithmetic progression, ensuring that all the previous $\sqrt{(n-2)/3}$ angles are included in the new arithmetic progression. Thus in this construction, the number of angles formed with points $A$ and $B$ is exactly $2\sqrt{(n-2)/3}-1$. 
\end{proof}

\subsection{Pinned Center Point}
We now move to pinning a single point. Here we consider a pinned center point and prove Theorem \ref{thm2} on the minimum number of distinct angles which have a given point $A$ as center point. 

Note first that in two dimensions, the answer is $\Theta(n)$. Lemma \ref{lem1} gives us the $\Omega(n)$ lower bound. To get the upper bound, imagine that the pinned point is the origin. We can place the remaining points in the plane such that their polar angles form an arithmetic progression, thereby having only $O(n)$ distinct angles with the pinned center point. We may vary the distances of these points from the origin so that they remain in general position.

In three dimensions, Lemma \ref{lem1} does not hold, but we can transform the problem into one of distinct distances:

\begin{lem}\label{pinnedcenter}
Consider a configuration of $n$ points in general position in three dimensions, and pin a point $A$. Then the number of distinct angles with $A$ as a center point is equal to the number of distinct distances of the projections of the other points onto a sphere centered at $A$.
\end{lem}
\begin{proof}
To count the number of angles of the form $\angle BAC$, where $A$ is pinned and $B$ and $C$ are any other distinct points from the set, note first that the distance from $A$ to $B$ or $C$ is irrelavent when determining the angle. Thus we can transform any point configuration in general position to one in which every point besides $A$ lies on a unit sphere centered at $A$ by replacing each point $P$ with a point $P'$ that lies at the intersection of the ray from $A$ through $P$ and the unit sphere. The general position prohibition against any three points on a line guarantees that for any distinct points $P$ and $Q$, $P'$ and $Q'$ are distinct. In our transformed construction, each angle $\angle B'AC'$ corresponds to a great-circle distance along the surface of the sphere, meaning the minimum achievable number of distinct angles is exactly equal to the minimum number of distinct distances on a sphere. 
\end{proof}

The best known upper bound for distinct distances on a sphere is $O(n)$, which is obtained by evenly distributing the points along any circle on the sphere. We cannot have $n$ points lie on a circle, but recall that we may vary the distances of the points from $A$ to create a legal configuration with $O(n)$ distinct angles with $A$ as the center point.

The best known lower bound for this problem, similar to the result for distinct distances in the plane by \cite{GuKa}, is a constant times $n/\log n$ (\cite{Tao}). This finishes the proof of Theorem \ref{thm2}.

A long-standing conjecture (discussed for example by \cite{EHP} and \cite{IoRu}) is that in fact there must be $\Omega(n)$ distinct distances for a configuration of $n$ points on the sphere. Still, the gap between the lower and upper bounds on this problem is rather small.

Theorem \ref{thm2} immediately allows us to write the following.

\begin{cor}
$\Ag{3} = \Omega(n/\log n)$.
\end{cor}

We could not find a better lower bound on $\Ag{3}$, despite the fact that pinning the center point of our angles is a large restriction on the angles we are considering.

\subsection{Pinned Endpoint}

The pinned endpoint case is quite different; there is no clear equivalence to a distinct distance problem. Here, we consider angles of the form $\angle ABC$ for a special point $A$, where $B$ and $C$ can be chosen freely from the $n-1$ remaining points.

In two dimensions, Lemma \ref{lem1} again gives us a lower bound of $\Omega(n)$ on the minimum number of distinct angles with a fixed endpoint. With regard to an upper bound on this minimum number, in any configuration, there are $O(n^2)$ angles formed with $A$ as a pinned endpoint since there are only $\binom{n-1}{2}$ choices for the other two points. No nontrivial upper bound is known.

In three dimensions, we have the lower bound stated in Theorem \ref{thm3}, which we repeat here for convenience:

\pinnedend*
\begin{proof}
The proof is very similar to that of Theorem \ref{thm1}. Fix a point $B$ (in addition to the pinned point $A$). In our proof of Theorem \ref{thm1}, we focused on angles that have center point $A$ and angles that have center point $B$, which leads us to consider the intersection of two cones. Instead, we now  focus on angles that have endpoint $A$ and center point $B$ and angles that have $A$ and $B$ as the two endpoints. As discussed in Section \ref{conestoruses}, this leads us to consider a cone and a spindle torus, both with axis $AB$, the intersection of which is again a circle (see Figure \ref{conetorusintersection}).

The rest of the proof continues in the same manner as the proof of Theorem \ref{thm1}: there can only be three points on any intersection of a particular cone with a particular spindle torus. So, if $x$ is the number of distinct cones formed and $y$ is the number of distinct spindle tori formed, we have $n \leq 3xy+2$. The number of distinct angles with $A$ as one of the endpoints is at least $\max(x,y)$, which (under the constraint that $xy \geq (n-2)/3$) is minimized when $x=y=\sqrt{(n-2)/3}$.
\end{proof}

\begin{figure}[htbp]
    \centering
    \includegraphics[scale=0.74]{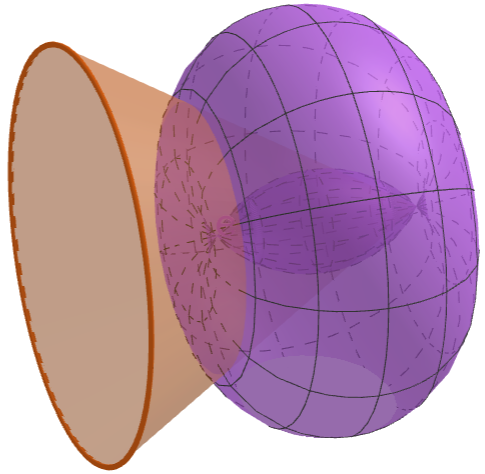}
    \caption{The intersection of a cone and spindle torus that share the same axis is a circle.}
    \label{conetorusintersection}
\end{figure}

 While Theorem \ref{thm3} provides a new and nontrivial lower bound, it seems intuitive that a much higher lower bound would hold; indeed, in the proof of the theorem, we did not consider any angles that were not formed with the fixed point $B$. In the absence of any known construction in three dimensions with fewer than the trivial order $n^2$ distinct angles with a pinned endpoint, we therefore conjecture the following.

\begin{conj}\label{conj2}
For any configuration of $n$ points in general position in two or three dimensions, the number of distinct angles formed with a pinned point $A$ as one of the endpoints is $\Theta(n^2)$.
\end{conj}

Conjecture \ref{conj2} clearly implies, and is a much stronger conjecture than, Conjecture \ref{conj1}.

\section{Constructions}\label{constructions}
\cite{FlKo} use geometric properties of the logarithmic spiral to construct a set of points in $\mathbb{R}^2$ with $O(n^2)$ distinct angles. This construction avoids the use of projections from hypercubes or hyperspheres which previously yielded the minimum number of distinct angles in general position in the plane. In the following constructions, we reuse geometric properties of the logarithmic spiral in $\mathbb{R}^3$. Note that as general position permits all points lying on a plane, the logarithmic spiral embedded into $\R^3$ is also a construction that has $O(n^2)$ distinct angles in three dimensions. We provide two new constructions, namely the cylindrical helix and the conchospiral, that use properties similar to those of the logarthmic spiral but that do not lie on any plane.

\begin{prop}[Cylindrical Helix]\label{cyl}
Let $\mathcal P = \{(\cos(2\pi j/n), \sin(2\pi j/n), j/n): j\in\{1,2,\ldots,n\}\}$. Then $\mathcal P$ is in general position and has $O(n^2)$ angles.
\end{prop}

\begin{proof}
    Notice that any line passes through a cylinder at most twice, so no three points in $\mathcal{P}$ lie on a line. Upon some inspection, no four points lie on a plane (and thus no four points lie on a circle). This is because a plane has the form $ax+by+cz=d$, or, rearranging, $xa + yb - d = -cz$. Plugging in $x = \cos(2\pi j/n), y = \sin(2\pi j/n), z = j/n$ for four different values of $j$, we see that we are trying to solve a system of four linear equations with only three variables. Then it suffices to notice that no two values of $j$ yield linearly dependent equations. Therefore, $\mathcal{P}$ is in general position.
    
    Next, we show that $\mathcal{P}$ yields $O(n^2)$ angles. The proof is similar to the proof by \cite{FlKo} that the logarithmic spiral has $O(n^2)$ angles. We have that $(\cos (t), \sin (t), t)$ is the parameterization of the cylindrical helix, $\mathcal{C}$. We consider the mappings $F_\alpha: \mathcal{C} \rightarrow \mathcal{C}$ given by
    \begin{equation}
        F_\alpha(\cos (t), \sin(t), t) = (\cos(t + \alpha), \sin (t + \alpha), t + \alpha)
    \end{equation}
    If we put this parameterization in cylindrical coordinates, we see that $\theta$ is mapped to $\theta + \alpha$, and $z$ is mapped to $z + \alpha$. Therefore, $F_\alpha$ is a rotation by $\alpha$ and a translation upwards, which maps triangles to similar triangles. Hence, $F_\alpha$ preserves angles. 

    Let $p_{j} = (\cos(2\pi j/n), \sin(2\pi j/n), j/n)$ and consider the triple $t = (p_{j_1}, p_{j_2}, p_{j_3})$. Let $m = \min \{j_1, j_2, j_3 \}$. Then, we have that $f_t \coloneqq F_{\frac{2\pi}{n}(1 - m)}$ maps $t$ to a triple with the same angles, with one of the points as $(\cos(2\pi/n), \sin(2\pi /n), 1/n)$. Hence, all angles in $\mathcal{P}$ can be formed with $(\cos(2\pi/n), \sin(2\pi /n), 1/n)$ as one of the points.

    Therefore, we have that each distinct angle in $\mathcal{P}$ can be formed by using $(\cos(2\pi/n), \sin(2\pi /n), 1/n)$ as one of the points. So, as there are $\binom{n-1}{2}$ ways to choose the other two points, and $3$ angles can be formed with a triple, then the number of distinct angles in $\mathcal{P}$ is at most $3 \dbinom{n-1}{2}$.
\end{proof}

\begin{remark}
    Due to the vertical symmetry of the cylindrical helix, the angles formed by $t = (p_{j_1}, p_{j_2}, p_{j_3})$ are the same as the angles formed by $t^\prime = (p_{n+1 - j_1}, p_{n+1 - j_2}, p_{n+1 - j_3})$. Let $m^\prime = \max\{j_1, j_2, j_3 \}$ and let $g:\mathcal{C} \rightarrow \mathcal{C}$ be such that $g(t) = t^\prime$. Then, the map $f_t^\prime \coloneqq F_{\frac{2 \pi}{n} (1 - m^\prime)} \circ g$ takes $t$ to a triple with $(\cos(2\pi/n), \sin(2\pi /n), 1/n)$ as one of the points. Thus, when $f_t \neq f_t^\prime$, there are two such triples formed with $(\cos(2\pi/n), \sin(2\pi /n), 1/n)$ as one of the points that yield the same angles. Thus the number of distinct angles formed by $\mathcal P$ is asymptotically $\dfrac{3}{2}\dbinom{n-1}{2}$, a factor of $1/2$ better than the logarithmic spiral.
\end{remark}

Next, we introduce another construction that produces $O(n^2)$ distinct angles where points are placed on a conchospiral.

\begin{prop}[Conchospiral]
    Let $\mathcal P = \{(e^{\beta j} \cos (\beta j), e^{\beta j} \sin (\beta j), e^{\beta j}): j\in\{1,2,\ldots,n\}\}$, where $\beta$ is a sufficiently small constant. Then $\mathcal P$ is in general position and has $O(n^2)$ angles.
\end{prop}

\begin{proof}
    The projection of the conchospiral onto the $(x,y)$ plane (or analagously the $(r,\theta)$ plane) is the logarithmic spiral. Therefore, as in \cite{FlKo}, by choosing $\beta$ sufficiently small, no three points of $\mathcal{P}$ lie on a line, and no four points of $\mathcal{P}$ lie on a circle.

    Next, let $S$ be the conchospiral, which has parameterization $(e^t \cos t, e^t \sin t, e^t)$. As in \cite{FlKo}, let $F_\alpha: S \rightarrow S$ be the set of mappings
    \begin{equation}
        F_\alpha(e^t \cos t, e^t \sin t, e^t) = (e^{t + \alpha} \cos (t + \alpha), e^{t + \alpha} \sin (t + \alpha), e^{t + \alpha})
    \end{equation}
    By putting this in cylindrical coordinates, we see that $F_\alpha$ is a rotation by $\alpha$ and a dilation by $e^\alpha$. Hence, $F_\alpha$ maps triangles to similar triangles and thus preserves angles. 
    
    Let $p_{j} = (e^{\beta j} \cos (\beta j), e^{\beta j} \sin (\beta j), e^{\beta j})$ and consider the triple $t = (p_{j_1}, p_{j_2}, p_{j_3})$. Let $m = \min \{j_1, j_2, j_3 \}$. Then, we have that $f_t \coloneqq F_{\beta(1-m)}$ maps $t$ to a triple with the same angles, with one of the points as $(e^{\beta} \cos (\beta), e^{\beta} \sin (\beta), e^{\beta})$. Hence, all angles in $\mathcal{P}$ can be formed with $(e^{\beta} \cos (\beta), e^{\beta} \sin (\beta), e^{\beta})$ as one of the points.

    Therefore, we have that each distinct angle in $\mathcal{P}$ can be formed by using $(e^{\beta} \cos (\beta), e^{\beta} \sin (\beta), e^{\beta})$ as one of the points. So, as there are $\binom{n-1}{2}$ ways to choose the other two points, and $3$ angles can be formed with a triple, then the number of distinct angles in $\mathcal{P}$ is at most $3 \dbinom{n-1}{2}$.
\end{proof}

\subsection{Self-Similarity}
We define self-similarity in the following way:
\begin{defn}
A point configuration $\mathcal P$ exhibits \emph{self-similarity} if there exists a point $A\in \mathcal P$ such that any angle formed from three points in the configuration can also be formed with $A$ as one of the points. That is, for any $B, C, D\in \mathcal P$, there exist $E, F\in \mathcal P$ such that $\angle BCD = \angle AEF$ or $\angle BCD = \angle EAF$. The point $A$ is called the \emph{point of self-similarity}.
\end{defn}

Both the configurations discussed in this section, as well as the logarithmic spiral construction from \cite{FlKo}, have self-similarity; the projections from the hypercube or hypersphere (discussed by \cite{FlHu} and \cite{FlKo}, respectively) do not. In fact, any point configuration that exhibits self-similarity has at most $3\dbinom{n-1}{2} = O(n^2)$ distinct angles since all the angles can be formed by choosing two points besides $A$ and choosing one of the three angles in the triangle formed by those two points and $A$.

Self-similarity seems to be an efficient way to minimize the number of angles; both the logarithmic spiral as well as the two three-dimensional constructions presented in this section employ this tool. This suggests the following conjecture:
\begin{conj}\label{conj self-sim}
    For $n$ sufficiently large, the configuration of $n$ points in general position with the smallest possible number of distinct angles exhibits self-similarity.
\end{conj}

If Conjecture \ref{conj self-sim} is true, to prove Conjecture \ref{conj1} it suffices to show that any self-similar configuration can't have any additional ways to reuse enough angles to lower the order of the number of distinct angles.

\section{Distinct Angle Chains}\label{angle chains}
Having examined bounds on the number of individual distinct angles that appear in various settings and constructions, we now turn our attention to chains of angles. We adapt the following definitions from \cite{PSW}. 

\begin{defn}
Given a $k$-tuple of angles $(\alpha_1,\ldots,\alpha_k)$, a \emph{$k$-chain} of that type is a $(k+2)$-tuple of points $(x_1,\ldots,x_{k+2})$ such that $\angle x_ix_{i+1}x_{i+2}=\alpha_i$ for all $i=1,\ldots,k$. We call two $k$-chains \emph{distinct} if they have different types. We let $L_k^{(d)}(n)$ denote the minimum number of distinct $k$-chains in a configuration of $n$ points in general position in $\R^d$. That is, 
\[L_k^{(d)}(n) = \min_{\mathcal P} |\{(\angle x_1x_2x_3, \angle x_2x_3x_4, \ldots, \angle x_k x_{k+1}x_{k+2}): x_i\in \mathcal P\}|,\] where the minimum is over configurations $\mathcal P$ of $n$ points in general position in $\R^d$.
\end{defn}

Note that $L_1^{(d)}(n) = \Ag{d}(n)$ by definition.

\subsection{Distinct Angle Chains in $\mathbb{R}^2$}
In two dimensions, we have the lower bound stated in Theorem \ref{lowerbound2d}:
\lowertwod*

\begin{proof}
The proof follows by induction on $k$. For the base case of $k=1$, this is just the definition of $\Ag{2}(n)$. Assume the result holds for $k-1$, and consider $k$-chains. Fix the first $k-1$ angles in the chain; by the induction hypothesis there are $\Omega(\Ag{2}(n)\cdot n^{k-2})$ ways of doing so. Once we fix these angles, there is at least one chain of points $(x_1, x_2 \ldots, x_k, x_{k+1})$ with those angles. Fix such a chain.

Now, note that there are $\Omega(n)$ choices of angle formed with $x_k$, $x_{k+1}$, and the final point of the $k$-chain, $x_{k+2}$. This follows directly from Lemma \ref{lem1}.

Thus in total there are $\Omega(\Ag{2}(n)\cdot n^{k-2}\cdot n) = \Omega(\Ag{2}(n)\cdot n^{k-1})$ distinct tuples of angles $(\alpha_1, \ldots, \alpha_{k})$ with associated chains, and the induction is complete.
\end{proof}

Next, we prove the upper bound on $L_k^{(2)}(n)$ stated in Theorem \ref{upperbound2d}.

\uppertwod*

\begin{proof} 
Consider the point set on the logarithmic spiral discussed by \cite{FlHu}; that is, in polar coordinates, $\mathcal{P} = \{ (e^{j \beta}, j\beta) : j \in [n] \}$. We label these points $p_j$ such that $p_j = (e^{j \beta}, j \beta)$. Define $d_p (p_i, p_j) = j-i$. Then, recall the special property of the logarithmic spiral that rotating points a constant angle along the spiral does not change the angle between the points. Thus, the angle tuple $(\alpha_1, \ldots, \alpha_k)$ corresponding to a chain $(x_1,x_2,\ldots, x_{k+2}) = (p_{j_1}, p_{j_2}, \ldots, p_{j_{k+2}})$ is repated with chains $(p_{j_1+c}, p_{j_2+c},\ldots, p_{j_{k+2}+c})$ for any integer constant $c$. In other words, the tuple of angles is determined by the $(k+1)$-tuple of values $(\ell_1, \ldots, \ell_{k+1})$ such that $\ell_i = d_p(p_{j_{i+1}}, p_{j_i})$. Each $\ell_i$ can at most range from $-(n-1)$ to $n-1$. Thus we have that the number of distinct angle $k$-chains in this configuration is at most
\[ |\{(\ell_1, \ldots, \ell_{k+1}) : -(n-1) < \ell_i < n-1 \}| = (2n-1)^{k+1} = O(n^{k+1}).\]
\end{proof}

Notice that the lower and upper bounds given in Theorems \ref{lowerbound2d} and \ref{upperbound2d} are only a factor of $n$ apart for \emph{any} value of $k$. Further, if Conjecture \ref{conj1} holds, the lower and upper bounds agree up to constant factors.

\subsection{Distinct Angle Chains in $\mathbb{R}^3$}
Now, let $E$ be a set of $n$ points in $\mathbb{R}^3$ with no three points on a line and no four points in a circle. In this setting, we were unable to find a construction with fewer than $c\cdot n^{k+1}$ distinct angle $k$-chains as obtained in $2$ dimensions (see Theorem \ref{upperbound2d}). Any lower bound on the number of distinct angle chains is more tricky to obtain than in the two-dimensional case because the argument used in the proof of Theorem \ref{lowerbound2d} no longer applies. We instead have the following:

\lowerthreed*
\begin{proof} This follows again from induction on $k$. There are three base cases. When $k=1$, this follows from our lower bound of $\Omega(n/\log n)$ on $\Ag{3}$. 

The $k=2$ case is more tricky. We aim to show an $\Omega(n)$ lower bound in this setting. 
We are looking at the minimum number of distinct pairs $(\alpha_1, \alpha_2)$ such that there is $(x_1, x_2, x_3, x_4) \in E^4$ with $\angle x_1 x_2 x_3 = \alpha_1$ and $\angle x_2 x_3 x_4 = \alpha_2$.

Fix the middle two points $x_2$ and $x_3$. It remains to choose points $x_1$ and $x_4$, where we will consider angles $\angle x_1x_2x_3$ as $\alpha_1$ and $\angle x_2x_3x_4$ as $\alpha_2$. 

For each point except $x_2$ and $x_3$ in the configuration, form two cones: one with vertex $x_2$ and one with vertex $x_3$, both with axis which is the line $x_2x_3$ and both which pass through the point. We see that the number of tuples $(\alpha_1, \alpha_2)$ is equal to the number of intersections of cones with vertex $x_2$ and cones with vertex $x_3$. This is because the cone with vertex $x_2$ determines $\alpha_1$, and the cone with vertex $x_3$ determines $\alpha_2$. 

Two cones that share an axis intersect at a circle, and we can only have three points per circle. Thus the number of relevant intersections of cones is at least $\frac{n-2}{3}$. Hence, the minimum number of distinct pairs $(\alpha_1, \alpha_2)$ such that there is $(x_0, x_1, x_2, x_3) \in E^4$ with $\angle x_0 x_1 x_2 = \alpha_1$ and $\angle x_1 x_2 x_3 = \alpha_2$ is $\Omega(n)$.

For $k=3$, first note that there are at least $\Ag{3}=\Omega(n/\log n)$ choices for the first angle. Once we fix the first angle, there is at least one triple of points $(x_1, x_2, x_3)$ that form this angle. Then, consider the third angle in the chain, $\angle x_3x_4x_5$. There must be $\Omega(\sqrt{n})$ choices for this angle, a result that follows immediately from our discussion on the pinned endpoint variant (see Theorem \ref{thm3}). Thus there are $\Omega(n^{3/2}/\log n)$ choices of angle chains of length $3$. 

Now we proceed to the induction step. Suppose the result holds for all values of $k$ up to but not including a certain value $k_0$. Consider a chain of length $k_0$, and fix all of the first $k_0-3$ angles. (By the induction hypothesis, the number of ways of doing this is at least the quantity given by the statement of the theorem, substituting $k=k_0-3$.) There exists at least one $(k_0-1)$-tuple $(x_1, \ldots, x_{k_0-1})$ with these fixed angles. Now, points $x_{k_0}$, $x_{k_0+1}$, and $x_{k_0+2}$ can be chosen among any of the remaining points; they are not confined in any way from the first $k_0-3$ angles. Thus there are at least $\Ag{3} = \Omega(n/\log n)$ choices for the last angle. 

This completes the induction. Indeed, we have shown that every time $3$ is added to $k$, the number of distinct angle chains is multiplied by $n/\log n$, and the statement of the theorem follows.
\end{proof}
\begin{remark}
    Notice that based on this proof, improvements to the lower bounds on $\Ag{3}$ or on the number of distinct angles with a pinned endpoint would immediately lead to improvements on this result. 
\end{remark}
\section{Loosening the General Position Restriction}\label{loosening}
We now turn to a variant of the two-dimensional distinct angle problem that eases the restriction on the maximum number of points on a circle or line. Permitting all $n$ points to lie on a circle or line leads to easy optimal (up to constant factors) constructions, so we do not want to completely discard the constraints on the points; we instead allow $O(\sqrt{n})$ points to be colinear or cocircular.

With this restriction, we can position the n points on $\sqrt n$ rays pointing out from the origin with polar angle $c \frac{2\pi} {\sqrt{n}} \ |\ c\in \{0,1,2,\ldots \sqrt{n}-1\}$. If we space the points linearly along each ray, so that the polar distances on a given ray are simply integers from 1 to $\sqrt n$, we get $n^{3/2}$ distinct angles with the origin as an endpoint. This is because without loss of generality, we can choose the center point to lie on the ray with polar angle 0, so we have $\sqrt n$ choices for the center point and then $n$ for the other endpoint. We can drastically improve on this by spacing the points exponentially along the rays instead of linearly. Our configuration, described in polar coordinates, is now the pinned origin plus all points of the form $(r,\theta)=(2^a,c \frac{2\pi} {\sqrt{n}})$ where $a \in \{0,1,2, \ldots \sqrt n -1\}$ and $c\in \{0,1,2,\ldots \sqrt{n}-1\}$. (See Figure \ref{sunshine}. All coordinates are in polar form throughout this discussion.)

Without loss of generality again, we can pick our center point to lie on the horizontal ray. Let us call this point B, meaning we have $B=(2^k, 0)$.
Now let $A=(1,0),$ $C=(2^\ell, c \frac{2\pi} {\sqrt{n}}),$ and $D=(2^{\ell+k}, c \frac{2\pi} {\sqrt{n}})$. 
Note now that $\angle OAC=\angle OBD$, where $O$ is the pinned origin. This follows from the fact that triangles $\Delta OAC$ and $\Delta OBD$ are similar, since they share a common angle and the ratios of the incident sides, namely $OD/OB$ and $OC/OA$, are both equal to $2^\ell$.
This means, then, that $\angle OBD$ depends only on $\ell$ and $c$, not on $k$. We have $\sqrt n$ choices for $\ell$ and $\sqrt n$ choices for $c$, giving us a total of $O(n)$ distinct angles. This is far better than we are able to do with a pinned endpoint under the requirements of general position, where we conjecture only $\Theta(n^2)$ distinct angles are achievable.

\begin{figure}[htbp]
    \centering
    \includegraphics[scale=0.6]{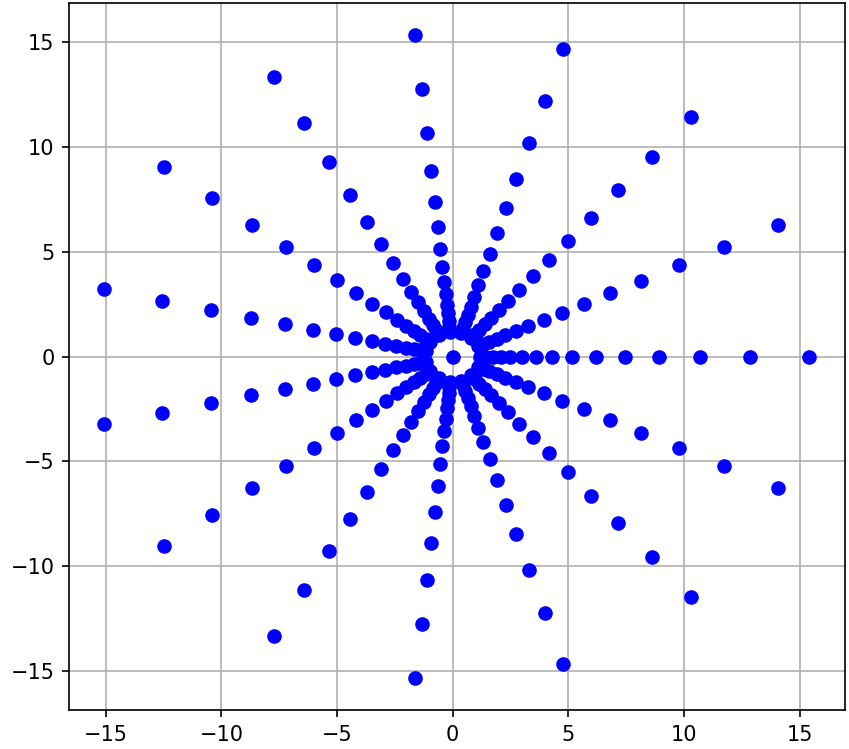}
    \caption{A configuration of $n$ points in $2$ dimensions with $O(\sqrt{n})$ points on any line or circle. Only $O(n)$ distinct angles are formed with the origin as one of the endpoints. }
    \label{sunshine}
\end{figure}
\section{Future Work}\label{future}
We discuss some possible directions for future research in this area.
\subsection{Connection to Incidence Problems in Algebraic Geometry}
One important target for future research is to decrease the gaps between the lower and upper bounds for both $\Ag{3}$ (currently $\Omega(n/\log n)$ and $O(n^2)$) and for the minimum number of distinct angles formed with a pinned endpoint (currently $\Omega(\sqrt{n})$ and $O(n^2)$). We describe an approach here that could in theory lead to results in these directions, though we have so far not been able to get improved bounds using these techniques.

\begin{defn}
For a point configuration $\mathcal P$ with $|\mathcal P| = n$, the \emph{energy} is given by \begin{equation} E(\mathcal P) = |\{(A,B,C,D,E,F)\in \mathcal P^6: \angle ABC = \angle DEF\}|.\end{equation}
\end{defn}
For any angle $\alpha$, let $N_\alpha(\mathcal P)$ denote the number of ordered triples in $\mathcal P$ forming angle $\alpha$:
\begin{equation} N_\alpha = |\{(A,B,C)\in\mathcal P^3: \angle ABC = \alpha\}|.\end{equation}
Note that the number of sextuples $(A,B,C,D,E,F)$ satisfying $\angle ABC = \angle DEF = \alpha$ is $N_\alpha^2$. So, denoting $\mathcal A$ to be the set of distinct angles formed by triples of points in $\mathcal P$, we may write
\begin{equation} E(\mathcal P) = \sum_{\alpha \in \mathcal A} N_\alpha^2.\end{equation}
We also have
\begin{equation}
    \sum_{\alpha\in \mathcal A} N_\alpha = |\mathcal P|^3 = n^3.
\end{equation}
We apply the Cauchy-Schwarz inequality:
\begin{align}
    \left(\sum_{\alpha \in \mathcal A} N_\alpha\right)^2 \leq \left(\sum_{\alpha \in \mathcal A} N_\alpha^2\right)\left(\sum_{\alpha \in \mathcal A} 1^2\right),
\end{align}
which tells us that
\begin{equation}
    |\mathcal A| \geq \frac{n^6}{E(\mathcal P)}.
\end{equation}
If we could show that $E(\mathcal P) = O(n^4)$ for any point configuration $\mathcal P$ in general position in three dimensions, then we would obtain the result that $\Ag{3} = \Theta(n^2)$. Even showing that $E(\mathcal P) = O(n^5)$ would already improve our lower bound on $\Ag{3}$.

We can turn this into an incidence problem as follows. Let each point $P$ be expressed as a vector with its three cartesian coordinates $\langle P_1, P_2, P_3\rangle$. We have that \[\angle ABC = \arccos\left(\frac{(A-B)\cdot (C-B)}{|A-B||C-B|}\right).\] Thus since $\arccos$ is bijective on its domain, $\angle ABC = \angle DEF$ if and only if \[\frac{(A-B)\cdot (C-B)}{|A-B||C-B|} =\frac{(D-E)\cdot (F-E)}{|D-E||F-E|},\] or, squaring both sides and rearranging,
\begin{equation}\label{weeeee}
\left((A-B)\cdot (C-B)\right)^2 |D-E|^2|F-E|^2 - \left((D-E)\cdot (F-E)\right)^2|A-B|^2|C-B|^2=0.
\end{equation}
We squared both sides so that the above expression is a polynomial (of degree $8$) in the $18$ variables $A_1, A_2, A_3, \ldots, F_1, F_2, F_3$, defining a ``nice" higher-dimensional surface.

We can now for example fix points $B$ and $E$ and ask how many quadruples of points $(A,C,D,F)$ satisfy Equation \ref{weeeee}. This allows us to think of the problem as having $n^2$ hypersurfaces (one for each $(B,E)$ pair) and $n^4$ quadruples of points, and the energy is the number of pairs of hypersurfaces and quadruples such that the quadruple of points is on the hypersurface. 

\subsection{Generalizing the General Position Requirement}\label{genpos}
We saw in Section \ref{loosening} that if we loosen the restrictions of general position to allow $O(\sqrt{n})$ points on a line or circle, then we can find a configuration in two dimensions that has $O(n)$ distinct angles with a pinned endpoint. More research can be done in this vein: specifically, if no point is pinned, what constructions minimize the number of distinct angles in this setting? Furthermore, what lower bounds do we have on the number of distinct angles with these constraints? We can also generalize this idea by allowing $O(n^\delta)$ points on any line or circle. How do all of these bounds change in three dimensions?

We can also go the other direction and further restrict the general position requirement. In two dimensions, the definition of general position is to have no three points on a line and no four points on a circle. In three dimensions, the classical definition of general position requires no four points on a plane and no five points on a sphere. These do not turn out to be natural conditions, so we instead chose to keep this definition as is. Indeed, the constructions shown in Section \ref{constructions} show that this stricter requirement does not prevent us from having configurations in $\R^3$ with $O(n^2)$ distinct angles; further, the stricter requirement does not immediately lead to any improvement on the lower bound on $\Ag{3}$.

We could, however, meaningfully change the question by requiring that there are only a constant number of points on any surface of dimension at most $2$. This would prohibit placing many points on a cone or spindle torus, therefore immediately producing a lower bound of $\Omega(n)$ distinct angles (in the same spirit as Lemma \ref{lem1}). On the other hand, this also prohibits all the explicit constructions that we considered: the logarithmic spiral, the cylindrical helix, and the conchospiral. The best construction of which we are aware that satisfies these stricter conditions is the projection of points from a hypersphere; see \cite{FlKo} for a discussion of the construction in two dimensions, where the number of distinct angles is $O\left(n^2 2^{22\sqrt{\log_2 n}}\right)$. (Projecting onto three dimensions has a similar outcome.) Seeing what other constructions arise and what nontrivial lower bounds on the number of distinct angles can be obtained in this setting is an interesting problem for future research.

\subsection{Distinct Angles on Surfaces}

A number of recent papers (see for example \cite{ShSo} and \cite{MaSh}) study the question of the minimum number of distinct distances on varieties of degree $2$ in $\R^3$. One interesting question would be to find bounds on the number of distinct angles among points on these general surfaces. However, care must be taken with the definition of ``angle" on these surfaces.

In the same vein, one can ask what is the minimum number distinct angles among points in the Poincare disk or Poincare ball. Distinct distances in hyperbolic surfaces have previously been studied by \cite{Meng}.

\nocite{*}
\bibliographystyle{abbrvnat}
\bibliography{DMTCS_angles}
\label{sec:biblio}

\end{document}